\documentclass[final,1p]{elsarticle}
\usepackage{pslatex}
\usepackage{amsfonts,xcolor,morefloats}
\usepackage{amssymb,amsmath,amsthm}
\newcommand{\C}{{\mathcal{C}}}
\newcommand{\mq}{\pmod{q}}
\newcommand{\mn}{\bmod{n}}

\newcommand{\lcm}{\mathrm{lcm}}

\newcommand{\ord}{\mathrm{ord}}
\newcommand{\PCL}{\mathrm{PCL}}
\newcommand{\CL}{\mathrm{CL}}
\newcommand{\cl}{\mathrm{cl}}
\newcommand{\NCL}{\mathrm{NCL}}

\newtheorem{theorem}{Theorem}
\newtheorem{lemma}[theorem]{Lemma}

\newtheorem{corollary}[theorem]{Corollary}

\newtheorem{proposition}[theorem]{Proposition}

\begin{document}
\begin{frontmatter}
\title{Dimensions of three types of BCH codes over $\text{GF}(q)$}
\author[cseust]{Hao~Liu}
\ead{hliuar@ust.hk}
\author[cseust]{Cunsheng~Ding}
\ead{cding@ust.hk}
\author[enu]{Chengju~Li}
\ead{lichengju1987@163.com} 
\address[cseust]{Department of Computer Science and Engineering, The Hong Kong University of Science and Technology, Clear Water Bay, Kowloon,
Hong Kong} 
\address[enu]{School of Computer Science and Software Engineering, East China Normal University, Shanghai, 200062, China} 

\begin{abstract}
BCH codes have been studied for over fifty years and widely employed in consumer devices, 
communication systems, and data storage systems. However, the dimension of BCH codes is 
settled only for a very small number of cases. In this paper, we study the dimensions of 
BCH codes over finite fields with 
three types of lengths $n$, namely $n=q^m-1$, $n=(q^m-1)/(q-1)$ and $n=q^m+1$. For narrow-sense primitive BCH codes with designed distance $\delta$, we investigate their dimensions for $\delta$ in the range $1\le \delta \le q^{\lceil\frac{m}{2}\rceil+1}$. For non-narrow sense primitive BCH codes, we provide two general formulas on their dimensions and give the dimensions explicitly in some cases. Furthermore, we settle the minimum distances of some primitive BCH codes. We also explore the dimensions 
of the BCH codes of lengths $n=(q^m-1)/(q-1)$ and $n=q^m+1$ over finite fields. 
\end{abstract}


\end{frontmatter}

\section{Introduction}

Throughout this paper, let $\text{GF}(q)$ be the finite field of order $q$, where $q$ is a prime power.
Let $n$ be a positive integer with $\gcd(n,q)=1$.
An $[n,k, d]$ linear code $\C$ over $\text{GF}(q)$ is a linear subspace of $\text{GF}(q)^n$ with dimension $k$ and
minimum (Hamming) distance $d$. Moreover, an $[n,k]$ linear code $\C$ is called \emph{cyclic} if
$(c_0, c_1, \ldots, c_{n-1}) \in \mathcal C$ implies
$(c_{n-1}, c_0, c_1, \ldots, c_{n-2}) \in \C$.
It is well-known that a cyclic code $\C$ over $\text{GF}(q)$ of length $n$ corresponds to an ideal of $\text{GF}(q)[x]/(x^n-1)$, i.e.,
$\mathcal C=\langle g(x) \rangle$, where $g(x)$ is a monic polynomial of the smallest
degree, $g(x)$ divides $(x^n-1)$ and is referred to as the \textit{generator polynomial} of $\C$.

Let $\alpha$ be a generator of $\text{GF}(r)^*$ and put $\beta = \alpha^{(r-1)/n}$, where $r=q^m$. Then $\beta$ is a primitive $n$-th root of unity. For any integer $i$ with $0\le i\le n-1$, let $m_i(x)$ denote the minimal polynomial of $\beta^i$ over $\text{GF}(q)$. For any integer $2\le \delta\le n$, define
$$g_{(q,n,\delta,b)} = \lcm \Big(m_b(x),m_{b+1}(x),\cdots,m_{b+\delta-2}(x)\Big)$$
where $b$ is an integer, lcm denotes the least common multiple of these minimal polynomials, and the addition in the subscript $b+i$ of $m_{b+i}(x)$ always means the integer addition modulo $n$. Let $\C_{(q,n,\delta,b)}$ denote the cyclic code of length $n$ over $\text{GF}(q)$ with generator polynomial $g_{(q,n,\delta,b)}(x)$. Then $\C_{(q,n,\delta,b)}$ is called a \emph{BCH code} of length $n$ and designed distance $\delta$.
When $b=1$, $\C_{(q,n,\delta,b)}$ is called a \emph{narrow-sense BCH code}. When $n=q^m-1$, $\C_{(q,n,\delta,b)}$ is called a \emph{primitive BCH code}. Furthermore, the set $[b, b+\delta-2]:=\{b, b+1, \ldots, b+\delta-2\}$ is called the \emph{defining set} of the BCH code $\C_{(q,n,\delta,b)}$.

BCH codes over finite fields are an important class of cyclic codes due to their error-correcting capability and efficient encoding and decoding algorithms, and are widely employed in compact discs, digital audio tapes and other data storage systems to improve data reliability.
Binary BCH codes were introduced by Hocquenghem  \cite{Hocquenghem1959}, Bose and Ray-Chaudhuri  \cite{Bose1960} in 1960s and were extended to general finite fields later \cite{GZ61}. Moreover, effective decoding algorithms have been developed for BCH codes, including the Peterson-Gorenstein-Zierler Algorithm and Berlekamp-Massey Algorithm, which facilitate the distribution of such codes.

Although BCH codes have been studied for decades, their parameters are seldom settled. 
So far, we have very limited knowledge on dimensions and minimum distances of BCH codes, in spite of some recent progress \cite{Ding2015a,Ding2015}. As pointed out by Charpin in \cite{Charpin1998}, their dimensions and minimum distances are difficult to determine in general.

Note that the BCH bound is naturally a lower bound on the minimum distances of the codes $\C_{(q,n,\delta,b)}$, i.e., $d \ge \delta$.
In this paper, we mainly focus on their dimensions.
Research into the dimensions of BCH codes began  as soon as BCH codes were discovered  \cite{MANN1962}.
 The dimensions of narrow-sense BCH codes were settled for 
$2\le \delta\le \min\{\lceil nq^{\lceil m/2\rceil}/(q^m-1)\rceil,n\}$ \cite{Dianwu1996, Aly2007}.
In addition,  the dimensions of the BCH codes $\C_{(q,n, \delta, b)}$ were investigated, where $\delta$ was among the first few largest coset leaders  \cite{Ding2016,Lia}.  Recently, the dimensions of some reversible BCH codes were
 studied in \cite{Ding2016a,Li2016,Lib}. For more information on the dimensions of BCH codes, we refer the reader to \cite{Ding2016a}.

In this paper, we extend earlier results and develop new ones on the dimension of BCH codes over finite fields. We investigate BCH codes with three types of lengths $n$, namely $n=q^m-1$, $n=(q^m-1)/(q-1)$ and $n=q^m+1$. For the primitive BCH codes, we give the dimensions of the narrow-sense BCH codes for $1\le \delta \le q^{\lceil\frac{m}{2}\rceil+1}$ and determine their minimum distances for a special case. We also provide two formulas on the dimension of $\C_{(n,q,\delta,b)}$ for the non-narrow-sense cases and settle the dimensions in some special cases. For $n=(q^m-1)/(q-1)$ and $n=q^m+1$, we determine the dimensions of BCH codes including some reversible cyclic codes, and extend some results of \cite{Ding2016a}.

\section{Some general results on the dimension of BCH codes} 

Throughout this paper, let $q$ be a prime power and $n>1$ be a positive integer with $\gcd(n, q)=1$. 
The $q$-adic expansion of an integer $a$ with $1\le a \le q^m-1$ is defined by $\sum_{i=1}^{m-1} a_iq^i$, where $0 \leq a_i <q$. The cardinality of a set $A$ by is denoted by $|A|$. Let $\mathbb{Z}_n=\{0,1,2,\cdots,n-1\}$ \linebreak[3] denote the ring of integers modulo $n$. For any $a\in\mathbb{Z}_n$, the $q$-cyclotomic coset of $a$ modulo $n$ is defined by
$$C_a=\{a,aq,aq^2,\cdots,aq^{l_a-1}\}\mn\subseteq \mathbb{Z}_n,$$
where $l_a$ is the least positive integer such that $aq^{l_a} \equiv a \pmod{n}$, and is the size of $C_a$. It is well known that $l_a\mid m$. The smallest element in $C_a$ is called the \textit{coset leader} of $C_a$ and denoted by $\cl(a)$. 

Let $[b, b+\delta-2]:=\{b, b+1, \ldots, b+\delta-2\}$ be the defining set of the BCH code 
$\C_{(q,n,\delta,b)}$.
It is easily seen that the dimension of $\C_{(q,n,\delta,b)}$ is given by 
\begin{equation}\label{EqDim}
\dim\left(\C_{(q,n,\delta,b)}\right) = n- \left| \bigcup_{a\in [b, b+\delta-2]}C_a \right|.
\end{equation}
Thus, to determine the dimension of the code $\C_{(q,n,\delta,b)}$, we need to find out all coset leaders of $C_a$ for $a \in [b, b+\delta-2]$ and the cardinalities of the cosets containing the coset 
leaders.

The following lemma and theorem were proved in \cite{Aly2007} and contain results in \cite{Yue2000,Dianwu1996} as special cases. 

\begin{lemma}\label{lem-AKS}
Let $n$ be a positive integer such that $\gcd(n, q)=1$ and $q^{\lfloor m/2 \rfloor}<n \leq q^m-1$, where 
$m=\ord_n(q)$. Then the $q$-cyclotomic coset $C_s=\{sq^j \bmod{n}: 0 \leq j \leq m-1\}$ has cardinality 
$m$ for all $s$ in the range $1 \leq s \leq n q^{\lceil m/2 \rceil}/(q^m-1)$. In addition, every $s$ with 
$s \not\equiv 0 \pmod{q}$ in this range is a coset leader.  
\end{lemma} 

\begin{theorem}\label{thm-AKS}
Let $n$ be a positive integer such that $\gcd(n, q)=1$ and $q^{\lfloor m/2 \rfloor}<n \leq q^m-1$, where 
$m=\ord_n(q)$. Then the narrow-sense BCH code $\C_{(q, n, \delta, 1)}$ of length $n$ and designed distance 
$\delta$ in the range 
$2 \leq \delta \leq \min\{\lfloor n q^{\lceil m/2 \rceil}/(q^m-1) \rfloor, n\}$ has dimension 
$$ 
k=n-m\lceil (\delta -1)(1-1/q) \rceil. 
$$
\end{theorem} 

It is necessary to make the following remarks on Theorem \ref{thm-AKS}. 
When $n=q^m-1$, Theorem \ref{thm-AKS} is quite useful, as 
$$
\min\{\lfloor n q^{\lceil m/2 \rceil}/(q^m-1) \rfloor, n\}=q^{\lceil m/2 \rceil},  
$$ 
which is large to an extent. When $n=q^\ell+1$, then $m=2\ell$ and 
$$
\min\{\lfloor n q^{\lceil m/2 \rceil}/(q^m-1) \rfloor, n\} =\frac{q^\ell}{q^\ell -1} < 2. 
$$
Hence, Theorem \ref{thm-AKS} is totally useless in the case that $n=q^\ell+1$ for any positive 
integer $\ell$. We will get back to Theorem \ref{thm-AKS} later in this paper.

\section{The primitive case that $n=q^m-1$}

In this section, we consider the primitive BCH codes of length $n=q^m-1$. This is the mostly studied 
case. Most of the references on primitive BCH codes focussed on the narrow-sense case \cite{MANN1962,Ding2015a,Ding2015,Ding2016,Ding2016a,Dianwu1996,Yue2000}, i.e., the case that $b=1$. 
The objective of this section is to determine the dimension of the primitive BCH code 
$\C_{(q, q^m-1, \delta, b)}$ for certain $\delta$ and $b$. Our results extend those 
in earlier references in the following two aspects: 
\begin{enumerate}
\item We consider also the case that $b \neq 1$.  
\item We investigate the dimension of the code $\C_{(q, q^m-1, \delta, 1)}$ for a larger range of 
$\delta$.  
\end{enumerate}
Throughout the whole section, $n=q^m-1$ unless otherwise stated. We denote by $\sum_{i=0}^{m-1}a_ip^i$ 
the $q$-adic expansion of an integer $a$ with $0 \leq a \leq q^m-1$.

\subsection{Auxiliary results about $q$-cyclotomic cosets modulo $n$.}
For the  primitive case, it was shown in   \cite{Yue2000,Dianwu1996,Aly2007} that for any integer $a$ with $1\le a \le q^{\lceil\frac{m}{2}\rceil}$ and $a\not\equiv 0\pmod{q}$, $a$ is a coset leader and $|C_a|=m$ (see Theorem \ref{thm-AKS}).
Here we consider only integers $a$ in the larger range $1\le a \le q^{\lceil\frac{m}{2}\rceil+1}$.

\subsubsection{The odd $m$ case}\label{sec-cosetleadersp1} 

Assume that $m\ge 3$ is an odd integer and let $h=(m-1)/2$.  Consider an integer $a$ with $q^{h+1}+1\le a \le q^{h+2}$ and $a\not \equiv 0 \pmod q$. Below we discuss the cardinality of $C_a$ and find out  some conditions under which $a$ is the coset leader of $C_a$.

When $m=3$, we have the following result.

\begin{lemma}
Let $m=3$. For $1\le a\le q^3-1$, the cyclotomic coset $C_a$ has cardinality $3$ except $C_{c(q^2+q+1)}$ for $1\le c\le q-1$, which are cosets of cardinality 1. Furthermore, $a\not\equiv 0\mq$ is a coset leader if and only if  $a_2<\min\{a_0,a_1\}$.
\end{lemma}

\begin{proof}
The proof is straightforward by analysing the $q$-adic expansion of $a$, and is omitted.
\end{proof}

We next consider the case that $m \ge 5$, and have the following.

\begin{lemma}\label{mOddCC}
Let $m\ge 5$ be an odd integer. Set $h=(m-1)/2$. For any integer $a$ with $q^{h+1}+1\le a \le q^{h+2}$, we have $|C_a|=m$.
\end{lemma}

\begin{proof}
Suppose that there exists an integer $a$ such that $q^{h+1}\le a \le q^{h+2}$ and $|C_a|\neq m$. If $m=5$ or $7$, then we must have $|C_a|=1$ since $|C_a|\mid m$. But from $a\le q^{h+2}$ we see that $a<qa<n$, which shows $|C_a|\ge 2$, a contradiction.

If $m\ge 9$, since $m$ is odd,\underline{} we have $l_a:=|C_a|\le m/3$ and $q^{l_a}a\equiv a\mn$.
Meanwhile, we have $a<q^{l_a}a\le q^{m/3}a<n$, which is a contradiction.
\end{proof}

The following lemmas give some necessary and sufficient conditions for an integer $a$ with $q^{h+1}\le a\le q^{h+2}$ and $a\not\equiv 0\mq$ to be a coset leader.

\begin{lemma}
Let $a$ be an integer with $q^{h+1}+1\le a\le q^{h+2}$ and $a\not\equiv 0\mq$. Denote the $q$-adic expansion of $a$ by $\sum_{i=0}^{h+1}a_iq^i$. If there exists an integer $r$ with $2\le r\le h-1$ such that $a_r \neq 0$, then $a$ is a coset leader with $|C_a|=m$.
\end{lemma}

\begin{proof}
To prove the desired conclusions, it suffices to show  that $q^ja\mn >a$ for all integers $1\le j\le m-1$.

Clearly, we have $a<q^ja<n$ for $1\le j\le h-1$.
When $h \le j\le m-r-1$, we have
\begin{align*}
q^j a\mn = \sum_{i=0}^{m-1}a_iq^{i+j}
\ge a_rq^{r+j}
\ge q^{h+2}
> a.
\end{align*}

When $m-r\le j\le m-1$,
\begin{align*}
q^j a\mn = \sum_{i=0}^{m-1}a_iq^{i+j}\ge a_0q^{m-r}\ge q^{h+2}> a.
\end{align*}
This completes the proof. 
\end{proof}

\begin{lemma}
Let $a$ be an integer with $q^{h+1}+1\le a\le q^{h+2}$ and $a\not\equiv 0\mq$. Denote the $q$-adic expansion of $a$ by $\sum_{i=0}^{h+1}a_iq^i$ and assume that $a_i=0$ for $2\le i\le h-1$.
If $a_1\neq 0$ and $a_h\neq 0$, then $a$ is a coset leader.

\end{lemma}
\begin{proof}
When $1\le j\le h$, since $h+j\le m-1$, we have $a<q^ja<n$.

When $h+1\le j\le m-2$,
\begin{align*}
q^j a\mn &= \sum_{i=0}^{m-1}a_iq^{i+j}
\ge a_1q^{j+1}
\ge q^{h+2}
> a.
\end{align*}

When $j=m-1$, $q^j a\mn\ge a_0q^{m-1}>a$.

Therefore, we have $q^ja\mn > a $ for $1\le j\le m-1$, which shows that $a$ is a coset leader.
\end{proof}

\begin{proposition}\label{thm2.5}
Let $a$ be an integer with $q^{h+1}+1\le a\le q^{h+2}$ and $a\not\equiv 0\mq$. Denote the 
$q$-expansion of $a$ by $\sum_{i=0}^{h+1}a_iq^i$ and assume that $a_i=0$ for $2\le i\le h-1$.

1) If $a_h=0$ and $a_1\neq 0$, then $ a$ is a coset leader if and only if $a_{h+1} \le a_1$.

2) If $a_1=0$ and $a_h\neq 0$, then $a$ is a coset leader if  and only if $a_{h+1}< a_0$.

3) If $a_1=a_h=0$, then $a$ is not a coset leader.
\end{proposition}

\begin{proof}
1) When $1\le j\le h-1$, we have $a<q^ja<n$.

When $j=h$, we have
\begin{align*}
q^j a\mn &= \sum_{i=0}^{m-1}a_iq^{i+j}
= a_1q^{h+1}+a_0q^{h}+a_{h+1}.
\end{align*}
Moreover, $a=a_{h+1}q^{h+1}+a_1q+a_0$. Then one can see that $a<q^ha \mn$ if and only if $a_{h+1}\le a_1$.

When $h+1\le j\le m-1$, $aq^j \mn \ge a_1q^{h+2}>a$.

 Therefore $q^ja \mn >a $ for $1\le j\le m-1$ if and only if $a_{h+1}\le a_1$ in this case.

2) can be proved similarly as  1).

3) For $a=a_{h+1}q^{h+1}+a_0$, since $q^ha\mn = a_0q^h+a_{h+1}<a$ we see that $a$ is not 
a coset leader.
\end{proof}

Summarizing the discussions above, we have the following conclusion.

\begin{proposition}\label{thmOddmPrimitive}
Let $m\ge 5$ be an odd integer and let $a$ be an integer with $1\le a \le q^{h+2}$ and $a\not\equiv 0\mq$. Then $|C_a|=m$ and $a$ is {\emph{not}} a coset leader for the following cases:

1) $a=a_{h+1}q^{h+1}+a_1q+a_0$, where $1\le a_{1}< a_{h+1}\le q-1$ and $1\le a_0\le q-1$;

2) $a=a_{h+1}q^{h+1}+a_hq^h+a_0$, where $1\le a_0\le a_{h+1}\le q-1$ and $1\le a_h\le q-1$;

3) $a=a_{h+1}q^{h+1}+a_0$, where $1\le a_0$ and $a_{h+1}\le q-1$.

Furthermore, 
$$|\{a: 1\le a \le q^{h+2},\, a\not\equiv 0\pmod q, \text{ and } a \text{ is not a coset leader }\}|=q(q-1)^2.$$ 
\end{proposition} 

The following follows from Proposition \ref{thmOddmPrimitive}.

\begin{corollary}
When $n=q^m-1$ and $m$ is odd, the smallest $a$ with $a\not\equiv 0\pmod q$ that is not a coset leader is $q^{(m+1)/2}+1$. 
\end{corollary}

\subsubsection{The even $m$ case}\label{sec-cosetleadersp2} 

For $m=2$ we have the following proposition, whose proof is straightforward and omitted here.

\begin{proposition}
Let $m=2$. For $1\le a\le n-1$ and $a \not\equiv 0\mq$, $a$ is a coset leader if and only if $a=a_1q+a_0$ with $a_1\le a_0$. Furthermore,
\begin{equation}
|C_a| = \begin{cases}
1,~a_0=a_1;\\
2,~a_0\neq a_1.
\end{cases}
\end{equation}
\end{proposition}

Next we consider the case that $m$ is even and $m \ge 4$. Set $h = m/2$. For an integer $a$ in the range $q^{\frac{m}{2}}\le a\le q^{\frac{m}{2}+1}$ with $a \not\equiv 0 \mq$, we have $a_0\neq 0$, $a_{h}\neq 0$. The following lemma concludes the cardinality of $C_a$ for $a$ being in this range.

\begin{lemma}\label{lammaPrimitiveMevenCC}
Let $m\ge 4$ be an even integer. Set $h=m/2$.
For $q^{h}\le a\le q^{h+1}$, we have
\begin{equation}
|C_a| = \begin{cases}
m/2, & \mbox{ if } a=c(q^h+1)~,1\le c\le q-1;\\
m, & \mbox{ otherwise.} 
\end{cases}
\end{equation}
\end{lemma}
\begin{proof}
Let $l_a=|C_a|$. It is well known that $l_a$ divides $m$. With similar arguments to those in  Lemma \ref{mOddCC} we can deduce that $l_a\ge m/2$, which implies $l_a = m$ or $m/2$.

Assume that $l_a=m/2=h$, which is equivalent to $q^ha\mn = a$. Comparing their $q$-adic expansions we have
\begin{align*}
q^ha\mn &= \sum_{i=0}^{m-1}a_iq^{i+h} = \sum_{i=h}^{m-1}a_{i-h}q^i + a_h = a_hq^h+\sum_{i=0}^{h-1}a_{i}q^{i} =a ,
\end{align*}
which is equivalent to $a_h=a_0$ and $a_i=0$ for $1\le i \le h-1$. The desired conclusion follows directly.
\end{proof}

It is known that $a$ is a coset leader for integers $1\le a \le q^h$ with $a\not\equiv 0\mq$ 
\cite{Dianwu1996}. Next we investigate the cosets $C_a$ with $q^h+1 \leq a \leq q^{h+1}$ and 
determine their coset leaders.

\begin{lemma} \label{lemmaME}
Let $m\ge 4$ be an even integer, and let $a$ be an integer with $q^h+1\le a\le q^{h+1}$ and $a\not\equiv 0\mq$.
If $a_r\neq 0$ for some $1\le r\le h-1$, then $a$ is a coset leader with $|C_a|=m$.
\end{lemma}
\begin{proof}
Again it suffices to show that $q^ja\mn >a $ for $1\le j\le m-1$. It is easy to check the following statements.
\begin{itemize}
  \item When $1\le j\le h-1$, $a<q^ja<n$.
  \item When $h\le j\le m-r-1$, $q^ja\mn \ge a_rq^{r+h} \ge q^{r+1} >a$.
  \item When $m-r\le j\le m-1$, $q^ja\mn \ge a_0q^{n-r} \ge q^{h+1} >a$.
\end{itemize}
This completes the proof. \end{proof}

\begin{lemma}
Let $a=a_hq^h+a_0$, where $1\le a_0$ and $a_h \le q-1$.
With the same assumptions on $m$ and $a$ in Lemma \ref{lemmaME}, then $a$ is a coset leader if and only if $a_0 \ge a_h$.
\end{lemma}

\begin{proof}
It is easy to see that $q^ja\mn >a $ when $1\le j\le h-1$ and $h+1\le j\le m-1$. Next we consider 
the case that $j=h$.
Assume $q^h a\mn < a$ and we have
\begin{align*}
q^ha\mn = a_0q^h+a_h < a_hq^h + a_0 = a,
\end{align*}
which is equivalent to $a_0<a_h$.
\end{proof}

Collecting the lemmas above, we arrive at  the following conclusion.

\begin{proposition}\label{thmEvenMPrimitive}
Let $m\ge 4$ be an even integer. Let $a$ be an integer with  $q^h+1\le a \le q^{h+1}$ and $a \not\equiv 0\mq$.

1) If $a= c(q^h+1)$ for some $c$ with $1 \le c \le q-1$, then $a$ is a coset leader with $|C_a|=m/2$.

2) If $a=a_hq^h+a_0$ with $1\le a_0 < a_h \le q-1$, then $a$ is \emph{not} a coset leader.

3) Otherwise, $a$ is a coset leader with $|C_a|=m$.

Furthermore, 
$$|\{a : q^h+1\le a \le q^{h+1}, \, a\not\equiv 0\pmod q, \text{ $a$ is not a coset leader}\}|=\frac{(q-1)(q-2)}{2}.$$ 
\end{proposition} 

The following is a consequence of Proposition \ref{thmEvenMPrimitive}. 

\begin{corollary}
When $n=q^m-1$ and $m$ is even, the smallest $a$ with $a\not\equiv 0\pmod q$ that is not a coset leader is $2q^{m/2}+1$. 
\end{corollary}

\subsection{Primitive BCH codes $C_{(q, m, \delta, b)} $ over GF($q$) with $\delta+b-2 \le q^{\lceil \frac{m+2}{2}\rceil}$}

With the conclusions on cyclotomic cosets developed in Sections \ref{sec-cosetleadersp1} and 
\ref{sec-cosetleadersp2}, we settle the dimension of the code $\C_{(n,q,\delta,b)}$ in some cases 
in this subsection.

\subsubsection{The case that $b=1$}

First we consider the narrow-sense BCH code, i.e., $b=1$. When $m$ is even, we have the following conclusion.

\begin{theorem}\label{thm-sum-primBCH1}
Let $m=2$ and $b=1$. For $2\le \delta \le n-1$, denote the $q$-adic expansion of  $\delta-1$ by $\delta-1 = a_1q+a_0$. Then the dimension of $\C_{(n,q,\delta,1)}$ is given by 
\begin{equation*}
\dim(\C_{(n,q,\delta,1)})=\begin{cases}
n-(2qa_1-a_1^2-1), ~&\text{if }a_0<a_1;\\
n-(2qa_1-a_1^2+2(a_0-a_1)), ~&\text{if }a_0\ge a_1.
\end{cases}
\end{equation*}

\end{theorem}

\begin{proof}
The desired conclusion follows directly from Proposition \ref{thmEvenMPrimitive}.
\end{proof}

\begin{theorem}\label{thmPrimitiveMeven}
Let $m$ be an even integer with $m\ge 4$. Set $h=m/2$. For $2\le \delta \le q^{m/2+1}$, denote the $q$-adic expansion of  $\delta-1$ by $\delta-1 = \sum_{i=0}^{h}\delta_iq^i$ and let $\delta_{Nq}=\delta-1-\lfloor(\delta-1)/{q}\rfloor$. We have then 
\begin{equation*}
\dim(\C_{(n,q,\delta,1)})=\begin{cases}
n-m\delta_{Nq}, \text{ if }\delta\le q^{h}+1;\\
n- \frac{m(2\delta_{Nq}-\delta_h^2)}{2}, \text{ if }\delta\ge q^{h}+2 \mbox{ and } \delta-1\ge \delta_h(q^{h}+1);\\
n-\frac{m\left(2\delta_{Nq}-(\delta_h-1)^2-2\delta_0\right)}{2}, \text{ if }\delta\ge q^{h}+2 
\mbox{ and } \delta-1< \delta_h(q^{h}+1). 
\end{cases}
\end{equation*}

\end{theorem}

\begin{proof}
By \eqref{EqDim}, the conclusion for $\delta \le q^h+1$ is obvious since all integers $a\le q^h$ satisfying $a\not\equiv 0\pmod q$ are coset leaders with $|C_a|=m$.

Now we assume that $\delta\ge q^h+2$.
  If $\delta-1\ge \delta_h(q^h+1)$,
  by Lemma \ref{lammaPrimitiveMevenCC}, we have
   $$|\{C_a: |C_a|=m/2, 1 \le a \le \delta-1\}|=\delta_h.$$
 It follows from Proposition \ref{thmEvenMPrimitive} that
 the integers $a$ with $q \nmid a$ that are not coset leaders are of the form
  $$a=c_hq^h+c_0 \text{ for } 1\le c_0<c_h\le \delta_h.$$
  It is easy to see that
  $$|\{a=c_hq^h+c_0: 1\le c_0<c_h\le \delta_h\}|=(\delta_h-1)\delta_h/2.$$ Then
  $$\left|\bigcup_{a\in [1, \delta-1]}C_a\right|=m(\delta_{Nq}-\delta_h-(\delta_h-1)\delta_h/2)+\delta_h\cdot m/2.$$
  Thus by \eqref{EqDim} the dimension of $\C$ is equal to
$$n-\left|\bigcup_{a\in [1, \delta-1]}C_a\right|= n-m(2\delta_{Nq}-\delta_h^2)/2.$$

If $\delta-1<\delta_h(q^h+1)$, i.e., $\delta_0\le \delta_h-1$, by Proposition \ref{thmEvenMPrimitive}, we similarly have
$$|\{C_a: |C_a|=m/2, 1 \le a \le \delta-1\}|=\delta_h-1$$ and
$$|\{a: 1 \le a \le \delta-1, q \nmid a, \text{ and $a$ is not a coset leader}\}|=(\delta_h-1)\delta_h/2-(\delta_h-1-\delta_0).$$
It then follows that
$$\left|\bigcup_{a\in [1, \delta-1]}C_a\right|=m\Big(\delta_{Nq}-(\delta_h-1)-\big((\delta_h-1)\delta_h/2-(\delta_h-1-\delta_0)\big)\Big)+(\delta_h-1)\cdot m/2.$$
Then by \eqref{EqDim} the dimension of the code
$\C$ is equal to
$$n-\left|\bigcup_{a\in [1, \delta-1]}C_a\right|
= n-m\left(2\delta_{Nq}-(\delta_h-1)^2-2\delta_0\right)/2.$$ 
\end{proof}

When $m$ is odd, the dimension of $\C_{(n,q,\delta,1)}$ is given as follows.

\begin{theorem}\label{thmDimensionModdPrimitve}
Let $m$ be an odd integer with $m \ge 5$. Put $h=(m-1)/2$ and $\delta_{Nq} = \delta-1-\lfloor\frac{\delta-1}{q}\rfloor$. For $2\le \delta \le q^{(m+3)/2}$, denote the $q$-adic expansion of  $\delta-1$ by $\delta-1 = \sum_{i=0}^{(m+1)/2}\delta_iq^i$ and $\delta_{Nq}=\delta-1-\lfloor\frac{\delta-1}{q}\rfloor$. We have then 
\begin{align*} 
\dim(\C_{(n,q,\delta,1)})= 
\begin{cases}
n-m\delta_{Nq}, \text{ if }\delta\le q^{h+1}+1;\\
n-m\left(\delta_{Nq}-(q-1)({\delta_{h+1}(\delta_{h+1}-1)}+\delta_1)-\delta_0\right), \\ 
 \ \ \ \ \ \text{ if }\delta\ge q^{h+1}+2, \delta-1< \delta_{h+1}(q^{h+1}+q);\\
n-m\left(\delta_{Nq}-(q-1)\delta_{h+1}^2\right), \\ 
\ \ \ \ \   \text{ if }\delta\ge q^{h+1}+2\text{ and } \delta_{h+1}(q^{h+1}+q)\le \delta-1< \delta_{h+1}q^{h+1}+q^{h};\\
n-m\left(\delta_{Nq}-(q-1)\delta_{h+1}^2-(\delta_h-1)\delta_{h+1}-\delta_0\right), \\ 
\ \ \ \ \    \text{ if }\delta\ge q^{h+1}+2 \text{ and }  \delta_{h+1}q^{h+1}+q^{h}\le \delta-1<\delta_{h+1}q^{h+1}+\delta_hq^{h}+\delta_{h+1};\\
n-m\left(\delta_{Nq}-(q-1)\delta_{h+1}^2-\delta_h\delta_{h+1}\right), \\ 
\ \ \ \ \   \text{   if }\delta\ge q^{h+1}+2\text{ and } \delta-1\ge \delta_{h+1}q^{h+1}+\delta_hq^{h}+\delta_{h+1},~\delta_h\ge 1.
\end{cases}
\end{align*}
\end{theorem}

\begin{proof}
With the help of Lemma \ref{mOddCC} and Proposition \ref{thmOddmPrimitive}, the proof is similar to that of Theorem \ref{thmPrimitiveMeven} and is omitted here.
\end{proof}

It should be pointed out that only the first conclusion in Theorem \ref{thmPrimitiveMeven} for 
the case that $\delta \leq q^h +1$ and the first conclusion in Theorem \ref{thmDimensionModdPrimitve} for the case $\delta \leq 2q^h +1$ were developed in \cite{Dianwu1996}. The rest of the conclusions in these two theorems are new. Clearly, Theorems \ref{thmPrimitiveMeven} and \ref{thmDimensionModdPrimitve} settle the dimension of the narrow-sense primitive BCH code 
$\C_{(n,q,\delta,1)}$ for $\delta$ in a much larger range.  

\par 

The minimum distances of these codes, fundamentally bounded by the BCH bound $d\ge \delta$, are very difficult to determine in general. It is known in the literature that $d=\delta$ for the narrow-sense BCH codes when $\delta\mid n$. Below we give a generalization of this conclusion.

\begin{lemma}\label{lemmaMinDis}
For any positive integer $n'$ with $(n',q)=1$ and $\gcd(n', q-1)=q-1$, let $\delta_b$ be an integer satisfying $\delta_b\mid \frac{n'}{q-1}$. Then for $\delta = k\delta_b$ with $1\le k\le q-1$, the minimum distance of the code $\C_{(n', q, \delta, 1)}$ is $\delta$.
\end{lemma}

\begin{proof}
Let $m=\ord_{n'}(q)$ and $\beta$ the $n'$th root of unity in GF$(q^m)$.
To deduce the desired result, we need to find a codeword with weight $\delta$ in this code. Denote
$$c(x) = \frac{x^{n'}-1}{x^{\frac{n'}{\delta_b}}-1} \times 
\prod_{i=1}^{k-1} \left(x^{\frac{n'}{\delta_b(q-1)}}-\beta^{\frac{in'}{q-1}}\right).$$

Note that $\beta^{in'/(q-1)} \in \text{GF}(q)$. Then it is clear that $c(x)\in \text{GF}(q)[x]$. Moreover, one can easily check that
$$c(\beta^j)=0 \text{ for } 1\le j\le \delta-1.$$ Thus $c(x) \in \C_{(n',q,\delta,1)}$.
It can be also checked that $w_H(c)\le k\delta_b=\delta$.
Meanwhile, $c(x)$ is not a zero codeword since $c(1) \neq 0$. Then the desired result follows from the BCH bound.
\end{proof}

The following theorem can be deduced from Theorem \ref{thmPrimitiveMeven} and Lemma \ref{lemmaMinDis}.

\begin{theorem}\label{thm-ding151}
Let $m=2h$, where $h$ is a positive integer. Then the primitive BCH code $\C_{(n,q,k(q^h+1),1)}$ has  parameters
$$\left[\,q^m-1, \, n-m(2k(q^h-q^{h-1})-(k-1)^2)/2,\, k(q^h+1)\,\right] $$
for $1 \le k \le q-1$.
\end{theorem}

As a special case of Theorem \ref{thm-ding151}, we have the following corollary.

\begin{corollary}
Let $n=q^2-1$. Then BCH code $\C_{(q^2-1,q,k(q+1),1)}$ has parameters
$$\left[~q^2-1,\,(q-k)^2,\,k(q+1)~\right] $$
for $1\le k\le q-1$.
\end{corollary}

\subsubsection{The case $b\ge 2$}

In this subsection, we will discuss the dimension of the BCH code $\C_{(n,q,\delta,b)}$ for $b\ge 2$. 
The dimension of the code $\C_{(n,q,\delta,b)}$ is more difficult to determine for $b \ge 2$. 
For convenience, we consider only the case that $m$ is odd. For even $m$, it can be similarly dealt with.

Assume $m$ is an odd integer.
We conclude a general dimension formula  of $\C_{(n,q,\delta,b)}$ for $b\ge 2$ in the following theorem.

\begin{proposition}\label{propositionBge2}
Let $m \ge 3$ be an odd integer. For integers $b,\delta$ with $1\le b\le n-1$ and $\delta+b-2\le q^{(m+3)/2}$, the dimension of $\C_{(n,q,\delta,b)}$ is given as follows.

1) When $b\le \lfloor\frac{b+\delta-2}{q}\rfloor$, we have $\C_{(n,q,\delta,b)} = \C_{(n,q,\delta+b-1,1)}$ and
$$\dim(\C_{(n,q,\delta,b)}) = n - m|\CL(1,b+\delta-2)|,$$ 
where 
$$\CL(b_1,b_2):=\{x\in [b_1,b_2]~ | ~|x| \text{ is a coset leader}\}.$$  
2) When $b\ge \lfloor\frac{b+\delta-2}{q}\rfloor+1$,
$$\dim(\C_{(n,q,\delta,b)}) = n - m(|\CL(b,b+\delta-2)|+|\PCL^+(b,b+\delta-2)|),$$
where $\PCL^+(b_1,b_2)$ denotes the set of positive pseudo coset leaders defined by
$$\PCL^+(b_1,b_2) = \{\,x\in \CL(1,b_1-1)~ | ~\exists j \text{ s.t. } (q^jx \mn) \in [b_1,b_2]\,\}$$ for positive integers $b_1,b_2$.
\end{proposition}

\begin{proof}
By Lemma \ref{mOddCC}, if $m \ge 5$ is odd, then we have $|C_a|=m$ for all $a$ with $1\le a\le q^{(m+3)/2}$. It follows from \eqref{EqDim} that the dimension is equal to
$$n-m|\{C_a: a\in [b,b+\delta-2]\}|.$$
 When $b\le \lfloor\frac{b+\delta-2}{q}\rfloor$, for any $a$ with  $1\le a\le b$, there exists an integer  $i$ such that $b\le \lfloor\frac{b+\delta-2}{q}\rfloor\le q^i a \le b+\delta-2$.
 This implies that 
 $$\{C_a: a\in [b,b+\delta-2]\}=\{C_a:a\in [1,b+\delta-2]\}, \text{ i.e., } \C_{(n,q,\delta,b)} = \C_{(n,q,\delta+b-1,1)}.$$
The desired conclusion on the dimension then follows.

When $b\ge \lfloor\frac{b+\delta-2}{q}\rfloor+1$, the desired result is straightforward from the definition of $\PCL^+$.
\end{proof}

Employing Proposition \ref{propositionBge2}, we can settle the dimension of the primitive BCH code 
for odd $m$ and $b \ge 2$. Below we consider two cases: $b+\delta-2=q^{(m+1)/2} \text{ or } q^{(m+3)/2}$.

Firstly, we let $b+\delta-2=q^{(m+1)/2}$. In this case, $b$ could be any integer between $2$ 
and $q^{(m+1)/2}$. Once $b$ is chosen, $\delta$ is fixed by $\delta=q^{(m+1)/2}-b+2$.    

\begin{theorem}\label{thm-sum-primBCH2}
Let $m\ge 3$ be an odd integer. Set $h=(m-1)/2$. 
For $1 \le b \le q^{h+1}$, let $b+\delta-2 = q^{h+1}$, i.e. $\delta=q^{h+1}-b+2$.

1) When $b\le q^{h}$, $\C_{(n,q,\delta,b)} = \C_{(n,q,q^{h+1},1)}$ and
$$\dim(\C_{(n,q,\delta,b)}) = n-m(q^{h+1}-q^{h}).$$

2) When $b\ge q^{h}+1$, we have
\begin{equation*}
\dim(\C_{(n,q,\delta,b)})=n-m(\delta-1).
\end{equation*}
\end{theorem}

\begin{proof}
1) The conclusion follows directly from Proposition \ref{propositionBge2} and Theorem \ref{thmPrimitiveMeven}.

2) By Proposition \ref{propositionBge2}, the desired conclusion can be drawn from the fact that 
$$|\PCL^+(b,q^{h+1})|=|[b,q^{h+1}]\setminus \CL(b,q^{h+1})|.$$ 
We prove this by giving a one-to-one correspondence between  
$$[b,q^{h+1}]\setminus \CL(b,q^{h+1})$$ and $\PCL^+(b,q^{h+1})$, which is $a\mapsto \cl(a)$. 
Recall that $\cl(a)$ denotes the coset leader of $C_a$.  
For any $a\in [b,q^{h+1}]$, we know that $a\in [b,q^{h+1}]\setminus \CL(b,q^{h+1})$ if and only if $a\equiv 0\mq$. Assume that $q^k||a$ for some positive integer $k$. Then we have $\cl(a) = a/q^k$ and $1\le \cl(a) \le q^{h} < b$. This shows that $\cl(a)\in \PCL^+(b,q^{h+1})$.

Furthermore, for any two integers $a_1,a_2\in [b,q^{h+1}]\setminus \CL(b,q^{h+1})$ with $a_1\neq a_2$, since $b\ge q^h+1$ we  must have $\cl(a_1)=a_1/q^{k_1}\neq a_2/q^{k_2} = \cl(a_2)$. Thus we find a one-to-one correspondence between $\PCL^+(b,q^{h+1})$ and $[b,q^{h+1}]\setminus \CL(b,q^{h+1})$, which completes the proof.
\end{proof}

Secondly, we consider the case: $b+\delta-2=q^{(m+3)/2}$. Similarly, in this case, $b$ could be any 
integer between $2$ and $q^{(m+3)/2}$.

\begin{theorem}\label{thm-sum-primBCH3}
Let $m\ge 5$ be an odd integer. Put $h=(m-1)/2$. 
For $1\le b\le q^{h+2}$, let $b+\delta-2 = q^{h+2}$, i.e., $\delta=q^{h+2}-b+2$.

1) When $b\le q^{h+1}+1$, $\C_{(n,q,\delta,b)} = \C_{(n,q,q^{h+2},1)}$ and 
$$\dim(\C_{(n,q,\delta,b)}) = n-m(q^{h+2}-q^{h+1}-q(q-1)^2).$$ 

2) When $b\ge q^{h+1}+2$,  for $b=kq^{h+1}+1$ with $1\le k\le q$ we have 
$$
\dim(\C_{(n,q,\delta,kq^{h+1}+1)})=q^m-1-m\left(q^{h+2}-kq^{h+1}-(q-k)^2q\right).
$$
\end{theorem}

\begin{proof}
1) The dimension follows directly from Proposition \ref{propositionBge2} and Theorem \ref{thmDimensionModdPrimitve}.

2) To determine the dimension of $\C_{(n,q,\delta,b)}$, where $b=kq^{h+1}+1$ with $1\le k\le q$, we need to compute $|\PCL^+(b,q^{h+2})|$.

 Define the set of non-coset-leaders in $[b,q^{h+2}]$ by $\NCL(b,q^{h+2})$, we have
$$\PCL^+(b,q^{h+2}) = \{\cl(a) \mid  a\in \NCL(b,q^{h+2})\}\cap [1,b-1].$$

By Theorem \ref{thmOddmPrimitive} we can divide the set $\NCL(b,q^{h+2})$ into the following four disjoint subsets:
$$\NCL_0:=\{a\in [b,q^{h+2}]\mid  a\equiv 0\mq\},$$
$$\NCL_1:=\{a=a_{h+1}q^{h+1}+a_hq^h+a_0\mid  k\le a_{h+1}\le q-1, 1\le a_0\le a_{h+1}, 1\le a_h\le q-1 \},$$
$$\NCL_2:=\{a=a_{h+1}q^{h+1}+a_1q+a_0\mid  k\le a_{h+1}\le q-1,1\le  a_1 < a_{h+1}, 1\le a_0\le q-1 \},$$
$$\NCL_3:=\{a=a_{h+1}q^{h+1}+a_0\mid  k\le a_{h+1}\le q-1, 1\le a_0\le q-1 \}.$$

Furthermore, define $J_i:=\{ \cl(a) \mid   a\in \NCL_i\}\cap [1,b-1]$ for $0\le i\le 3$. Then we have $$\PCL^+(b,q^{h+2}) =\cup_{i=0}^3 J_i.$$ We derive the cardinality of each $J_i$ as follows.

We first have 
$$|J_0|=|\NCL_0| = \lfloor (q^{h+2}-b+1)/q\rfloor = q^{h+1}-kq^h.$$ 

We can then check that 
$$J_1 = \{ a_0q^{h+1} +a_{h+1}q+a_h \mid  1\le a_0\le k-1,\, k\le a_{h+1}\le q-1,\, 1\le a_h\le q-1\},$$
which leads to $|J_1|=(k-1)(q-k)(q-1)$.

Similarly we can deduce $|J_2| = (k-1)(q-k)(q-1)$ and  $|J_3| = |\NCL_3| = (q-k)(q-1)$.
Next we analyse relations among the $J_i$s.

First we calculate $|J_0 \cap J_3|$. For any $a=a_{h+1}q^{h+1}+a_0\in \NCL_3$, from the proof of Theorem \ref{thm2.5} we can see $\cl(a) = a_0q^h+a_{h+1} <b$, which implies $\cl(a)\in J_3$. And it is easy to see that $$\cl(a)=a_0q^h+a_{h+1}\in J_0$$ if and only if $$b=kq^{h+1}+1\le q\cdot \cl(a)\le q^{h+2},$$ which is equivalent to $a_0\ge k$. Thus $$J_0 \cap J_3 = \{a_0q^h+a_{h+1}\mid k\le a_0,a_{h+1}\le q-1\}.$$

Then we show that $J_1\cap J_2 = \emptyset$. For any $a,a'$ that satisfy
$$a=a_{h+1}q^{h+1}+a_hq^h+a_0\in \NCL_1 \text{ and } a' = a'_{h+1}q^{h+1}+a'_1q+a'_0\in \NCL_2,$$ we have
$$ \cl(a) = a_0q^{h+1}+a_{h+1}q+a_h \neq a'_1q^{h+1} +a'_0q^h +a'_{h+1}= \cl(a')$$
for any $a_{h+1},a_h,a_0$ and $a'_{h+1},a'_h,a'_0$ in the definition of $\NCL_1$ and $\NCL_2$. Thus by the definitions of $J_{1},J_2$ we see $J_1\cap J_2=\emptyset$. Similarly, we can deduce that $J_1\cap J_3 = J_2 \cap J_3 = \emptyset$.

Lastly it is obvious that $J_0\cap J_1 = J_0\cap J_2 = \emptyset$ since for any $ a \in \NCL_{1}\cup \NCL_2$ we have $\cl(a)\ge q^{h+1}$ while for any $a \in \NCL_{0}$, we have $\cl(a)< q^{h+1}$.

Due to these relations among the $J_i$s, the cardinality of $\PCL^+$ becomes 
$$|\PCL^+(b,q^{h+2})| = |J_0| +  |J_1| + |J_2| + |J_3| - |J_0\cap J_3|.$$

Plugging the corresponding cardinalities into the formula above we obtain 
$$|\PCL^+(b,q^{h+2})| =q^{h+1}-kq^h+(k-1)(q-k)(2q-1).$$
Thus by Proposition $\ref{propositionBge2}$,
\begin{align*}
\dim(\C)  &= n - m(|\CL(b,q^{h+2})|+|\PCL^+(b,q^{h+2})|)\\
&= q^m-1-m\left(q^{h+2}-b+1-|\NCL(b,q^{h+2})|+|\PCL^+(b,q^{h+2})|\right)\\
&= q^m-1-m\left(q^{h+2}-kq^{h+1}-(q-k)^2q\right).
\end{align*}

\end{proof}

The proof above can also be smoothly applied to the case that $b=kq^{h+1}+l$, where $l\le k$. For other $b$, however, many miscellaneous details should be discussed and worked out, and we omit them here.

\subsubsection{The case that $b\le -1$}

For simplicity, we assume that $m \ge 5$ is odd. By Proposition \ref{thmOddmPrimitive},
for each integer $a$ with $1\le a \le q^{h+2}$ and $a\not\equiv 0\mq$, we have $|C_a|=m$.
 As before we set $h=(m-1)/2$. We provide a general formula on the dimension of the  BCH code $\C_{(n,q,\delta,b)}$ with $b\le -1$ in the following proposition.

\begin{proposition}\label{generalDimensionPrimitiveNegB}
Let $m\ge 3$ be an odd integer. Set $h=(m-1)/2$. For integers $b,\delta$ with $-q^{h+2}\le b\le -1$ and $1\le \delta+b-2\le q^{h+2}$, the dimension of $\C_{(n,q,\delta,b)}$ is given by 
\begin{equation*}
\dim(\C_{(n,q,\delta,b)}) = n-m(|\CL(b,b+\delta-2)|-|\PCL^-(b,b+\delta-2)|),
\end{equation*}
where $$\CL(b_1,b_2):=\{\,x\in [b_1,b_2]\mid  |x| \text{ is a coset leader\,}\}$$ and $\PCL^-(-b_1,b_2)$ denotes the set of the negative pseudo coset leaders, which is defined as 
$$\PCL^-(-b_1,b_2) = \{\,x\in \CL(-b_1,-1)\mid  \exists j \text{ s.t. } (q^jx \mn) \in [1,b_2]\}$$ for  positive integers $b_{1},b_2$.
\end{proposition}

\begin{proof}
Since $|C_a|=m$ for $a \in [b, b+\delta-2]$, the dimension of the code $\C_{(n,q,\delta,b)}$ is determined by the cardinality of the set
$$\{\,C_a \mid a\in[b,b+\delta-2]\,\}=  \{\,C_a \mid a\in[b,-1]\,\}\cup \{0\}\cup \{\,C_a\mid a\in[1,b+\delta-2]\,\}.$$

Clearly, $|\{\,C_a\mid a\in[1,b+\delta-2]\,\}| = |\CL(1,b+\delta-2)|$. Since $C_i=C_j$ is equivalent to $C_{-i}=C_{-j}$, we deduce that $|\{\,C_a\mid  a\in [b,-1]\} = |\CL(b,-1)|$ for $b\le -1$. Furthermore, by definition $$|\PCL^-(b,b+\delta-2)|=|\{\,C_a \mid a\in[b,-1]\,\}\cap  \{\,C_a\mid a\in[1,b+\delta-2]\,\}|.$$
Thus the dimension is given by
\begin{align*}
\dim(\C_{(n,q,\delta,b)})=&n-m|\{\,C_a \mid a\in[b,b+\delta-2]\,\}|\\
=&n-m\left(|\CL(b,-1)|+1+|\CL(1,b+\delta-1)|-|\PCL^-(b,b+\delta-2)|\right)\\
 =& n-m\left(|\CL(b,b+\delta-2)|-|\PCL^-(b,b+\delta-2)|\right).
\end{align*}
\end{proof}

By Proposition \ref{generalDimensionPrimitiveNegB}, to determine the dimension of $\C_{(n,q,\delta,b)}$ for negative $b$,  we need to calculate the cardinality of $\PCL^-(b,b+\delta-2)$.
The following lemma on $\PCL^-$ can be concluded from Lemma 2 of \cite{Li2016}.

\begin{lemma}\label{lemmaInLSX}
Let $m\ge 3$ be an odd integer. Let $h=(m-1)/2$. Then
$$\PCL^-(-q^{h+1},q^{h+1}) = \{-(q^{h+1}-u)\mid  1\le u\le q-1\}\cup \{(1-uq^h)\mid  1\le u\le q-1\}.$$ 
\end{lemma}

\begin{corollary}
Let $m\ge 5$ be an odd integer. Set $h=(m-1)/2$.
Then for $b,\delta$ with $q-q^{h+1}\le b<0<b+\delta-2\le q^{h+1}-q$, we have $$\PCL^-(b,b+\delta-2)=\emptyset.$$

\end{corollary}

With these results on $\PCL^-(b,b+\delta-2)$, we can now calculate the dimension of $\C_{(n,q,\delta,b)}$ in a few cases.
Since $C_i=C_j$ if and only if $C_{-i}=C_{-j}$ for any integers $i,j$, the dimensions of the BCH codes with defining sets $[-b_1,b_2]$ and $[-b_2,b_1]$ are the same.
 Therefore we consider only the case: $-b\le b+\delta-2$.

\begin{theorem}\label{thm-sum-primBCH4}
Let $m \ge 5$ be an odd integer. Set $h=(m-1)/2$ and $\delta_{Nq}= \delta-\lfloor{b}/{q}\rfloor-\lfloor(\delta+b-2)/q\rfloor-2$.  For $1\le -b\le b+\delta-2$, the dimension of $\C_{(n,q,\delta,b)}$ can be settled for the  following cases.

1) When $b+\delta-2 \le q^{h+1}-q$,
$\dim(\C_{(n,q,\delta,b)}) = n-m\delta_{Nq}-1.$

2) When $b+\delta-2= q^{h+1}$,

$$\dim(\C_{(n,q,\delta,b)}) =\begin{cases}
n-m\delta_{Nq}-1, \text{ if }-b<q^h-1;\\
n-m(\delta_{Nq}-\lfloor\frac{1-b}{q^h}\rfloor)-1, \text{ if }q^h-1\le -b\le (q-1)q^h-1;\\
n-m(\delta_{Nq}-q+1)-1, \text{ if }(q-1)q^h\le -b\le q^{h+1}-q;\\
n-m(\delta_{Nq}-q+1-l)-1, \text{ if }-b=q^{h+1}-q+l,~1\le l\le q-1.
\end{cases}$$
\end{theorem}

\begin{proof}
The conclusions follow directly from Proposition \ref{generalDimensionPrimitiveNegB} and Lemma \ref{lemmaInLSX}.
\end{proof}


In this section, we discussed the parameters of the primitive BCH code $\C_{(n,q,\delta,b)})$ with the defining set in the range $[-q^{\lceil{m}/{2}\rceil+1}, q^{\lceil{m}/{2}\rceil+1}]$. We found out all coset leaders in such range, and settled the dimensions of the narrow-sense BCH code for consecutive  
$\delta$ in the corresponding range. The minimum distances were also determined for a special class 
of $\delta$. We discussed also the cases of $b \neq 1$ and developed general formulas on the dimension for $b \ge 1$ and $b\le -1$ respectively. While it would be cumbersome to discuss all $b$ and $\delta$ in such non-narrow-sense cases, we considered a few cases where the dimensions could be determined. The discussions for even $m$ are left for future research.

\section{The projective case that $n=(q^m-1)/(q-1)$}

BCH codes $\C_{(n,q,\delta,b)})$ with length $n=(q^m-1)/(q-1)$ are called \textit{projective}. 
There may be only two references on projective BCH codes for $q>2$ \cite{Ding2016a,Lia}. In 
\cite{Ding2016a}, the 
dimension of the projective BCH code $\C_{(n,q,\delta, b)})$ is settled for even $m$ and some 
$b$ and $\delta$ being in certain range. 
The objective of this section is to complement the work of \cite{Ding2016a} by studying the 
dimension of $\C_{(n,q,\delta, b)})$ for odd $m$. 
Throughout this section, we always let $n=(q^m-1)/(q-1)$. 

\subsection{Auxiliary results about $q$-cyclotomic cosets modulo $n$}

Lemma 27 of \cite{Ding2016a} characterized all $q$-cyclotomic coset leaders modulo $n$ in the range $q^{(m-2)/2} \le a \le q^{m/2}$ for even $m$. In this subsection, we assume that $m$ is odd. By Theorem \ref{thm-AKS},  each integer $a$ with $1 \le a\le q^{(m-1)/2}$ is a coset leader with $|C_a|=m$. Thus, below we consider $a$ in the range $q^{(m-1)/2}\le a\le q^{(m+1)/2}$.

\begin{proposition}\label{thmOddMProjective}
Assume that $m$ is an odd integer with $m\ge 5$. Set $h=(m-1)/2$. Let $a$ be an integer with $q^{(m-1)/2}\le a\le q^{(m+1)/2}$ and $a\not\equiv 0\mq$. Then $|C_a|=m$ and $a$ is \emph{not} a coset leader in the following three cases:

1) $a=a_h\sum_{i=1}^h q^i+q+a_0$ with $q+a_0-2a_h\le 1 $;

2) $a=a_h\sum_{i=1}^hq^i+a_0$ with $a_h<a_0\le 2a_h$;

3) $a=a_hq^h+a_{h-1}\sum_{i=0}^{h-1}q^i+1$ with $a_h+a_{h-1}\ge q$ or $a_h+a_{h-1}=q-1$ while $2a_{h-1}\ge q$.
\end{proposition}

\begin{proof}
Let $h=(m-1)/2$. Denote the $q$-adic expansion of $a$ by $a=\sum_{i=0}^{h}a_iq^i$.  We have $a_h\neq0$  and $a_0\neq 0$ by assumption. As before, we consider $q^ja\mn$ for $1\le j\le m-1$ in the following cases.

Case 1: When $1\le j\le h-1$, it is clear that we have $a<q^ja<n$.

Case 2: When $j=h$, we have
\begin{align*}
q^ja\mn &= \sum_{i=h}^{m-2}(a_{i-h}-a_h)q^i - a_h\sum_{i=0}^{h-1} q^i.
\end{align*}

Case 2.1: If $a_{i-h}-a_h = 0$ for all $h\le i\le m-2$, then $a=a_h\sum_{i=0}^hq^i$ and
$$q^ha\mn = n-a_h\sum_{i=0}^{h-1} q^i = \frac{q^m-1-a_h(q^h-1)}{q-1}>a.$$

Case 2.2: If one of these $a_{i-h}-a_h$ is nonzero, let $k$ be the largest index such that $a_k-a_h\neq 0$. Suppose that $a_k-a_h<0$,  since $a_h-a_k-1\le q-2$,  we have
\begin{eqnarray*}
q^ha\mn &=& n - \sum_{i=h}^{k+h}(a_{h}-a_{i-h})q^i +a_h\sum_{i=0}^{h-1} q^i \\ 
&\ge& q^{m-1} - \sum_{i=h}^{k+h}(a_{h}-a_{i-h}-1)q^i+a_h\sum_{i=0}^{h-1} q^i \\
&\ge& q^{m-2} >a.
\end{eqnarray*}
Then we consider the case that $a_k-a_h > 0$. Note that $a_h<a_k\le q-1$, which  gives $a_h\le q-2$.

Case 2.2.1 : If $k\ge 2$, from $a_h\le q-2$ we have
\begin{eqnarray*}
q^ha\mn  &=& \sum_{i=h}^{h+k}(a_{i-h}-a_h)q^i - a_h\sum_{i=0}^{h-1} q^i \\
&\ge & q^{h+k}+\sum_{i=h}^{h+k-1}(a_{i-h}-a_h)q^i - a_h\sum_{i=0}^{h-1}q^i \ge q^{h+1} \ge a.
\end{eqnarray*}

Case 2.2.2 : If $k = 1$, then $a_i=a_h$ for $2\le i\le h-1$.
 Suppose that $q^ha\mn < a$, which is equivalent to
\begin{equation}\label{eq1}
(a_1-a_h)q^{h+1}+a_0q^h < \sum_{i=0}^{h}(a_i+a_h)q^i.
\end{equation}

If $a_1-a_h \ge 2$, then $a_h\le q-3$ and (\ref{eq1}) gives
$$2q^{h+1}+a_0q^h < (2q-6)q^{h+1}+\sum_{i=0}^{h-1}(a_i+a_h)q^i,$$
equivalently
$$(a_0+6)q^h < \sum_{i=0}^{h-1}(a_i+a_h)q^i < 2\sum_{i=1}^{h} q^i,$$
which is a contradiction.

If $a_1-a_h=1$, (\ref{eq1}) becomes
\begin{equation}\label{eq2}
(q+a_0-2a_h)q^h<\sum_{i=0}^{h-1}(a_i+a_h)q^i=2a_h\sum_{i=2}^{h-1} q^i+ (a_1-a_h)q+a_0-a_h.
\end{equation}
Since $a_h\le q-2$ and $a_i\le q-1$ for all $i$, we have
$$\sum_{i=0}^{h-1}(a_i+a_h)q^i\le (2q-3)\sum_{i=0}^{h-1}q^i=q^h+(q-2)\sum_{i=0}^{h-1}q^i+1.$$
Then (\ref{eq2}) holds only if $q+a_0-2a_h\le1$. If $q+a_0-2a_h\le 0$, the inequality clearly holds. Otherwise, (\ref{eq2}) holds if and only if $a_{h-1} +a_h \ge q$.

In conclusion, for the case $k=1$, i.e. $a_i=a_h$ for $2\le i\le h-1$, we have $q^ha\mn <a$  for the following cases:

a) $a_1-a_h =1 $ and $q+a_0-2a_h\le 0$;

b) $a_1-a_h =1 $, $q+a_0-2a_h=1$ and $a_{h-1} +a_h \ge q$.

Case 2.2.3 : If $k=0$, then $a_i=a_h$ for $1\le i\le h-1$, and $q^ja\mn<a$ is equivalent to

$$a_0q^h<\sum_{i=0}^h (a_i+a_h)q^i = 2a_h\sum_{i=0}^h q^i + a_0-a_h.$$

Since $a_h\le q-2$, this inequality is true if and only if $a_h< a_0\le 2a_h$.
Then we complete the discussion for $j=h$.

Case 3 :  When $j=h+1$, we have
$$q^ja\mn = \sum_{i=h+1}^{m-2} (a_{i-h-1}-a_{h-1})q^i - a_{h-1}\sum_{i=0}^{h-1}q^i +a_h.$$

If $a_{i-h-1}-a_{h-1} = 0$ for all $h+1\le i\le m-2$, then one can see that $q^ja\mn >a$.
Otherwise, let $k$ be the largest index such that $a_k-a_{h-1}\neq 0$. If $a_k-a_{h-1}< 0$, similar to Case 2.2 we can show that $q^ja\mn>a$. We then  consider the case that $a_k-a_{h-1}>0$. Assume that $q^ha\mn <a$, which is equivalent to
\begin{equation}\label{eq3}
\sum_{i=h+1}^{h+k+1} (a_{i-h-1}-a_{h-1})q^i - a_{h-1}\sum_{i=0}^{h}q^i +a_h<\sum_{i=0}^h a_iq^i.
\end{equation}

Case 3.1 :  If $k\ge 1$, due to the fact that $a_k-a_{h-1}>0$ we have $a_{h-1}\le q-2$, and then

$$\sum_{i=h+1}^{h+k+1} (a_{i-h-1}-a_{h-1})q^i - a_{h-1}\sum_{i=0}^{h}q^i +a_h \ge
2q^{h+1} -  a_{h-1}\sum_{i=0}^{h}q^i +a_h
\ge q^{h+1}>a.$$

Case 3.2 :  If $k=0$, we have $a_i=a_{h-1}$ for $1\le i\le h-2$. Then (\ref{eq3}) becomes
\begin{equation}\label{eq4}
(a_0-a_{h-1})q^{h+1} + a_h <(a_h+a_{h-1})q^h+2a_{h-1}\sum_{i=0}^{h-1}q^i +a_0-a_{h-1}.
\end{equation}
Since $a_h\le q-1$ and  $a_{h-1}\le q-2$, we deduce
\begin{align*}
(a_0-a_{h-1})q^{h+1} + a_h &< (a_h+a_{h-1})q^h+2a_{h-1}\sum_{i=0}^{h-1}q^i +a_0-a_{h-1}\\
&\le q^{h+1}+(q-2)\sum_{i=0}^{h-1}q^i+a_0-a_{h-1}-1   .
\end{align*}

If $a_0-a_{h-1}\ge 2$ the inequality would not survive.
Otherwise, we have $a_0-a_{h-1} = 1$, the inequality (\ref{eq4}) becomes
$$q^{h+1}+a_h-1<(a_h+a_{h-1})q^h+2a_{h-1}\sum_{i=0}^{h-1}q^i .$$

1) If $2a_{h-1}<q$, it holds if and only if $a_h+a_{h-1}\ge q$;

2) If $2a_{h-1}\ge q$, the inequality becomes
$$q^{h+1}+a_h<(a_h+a_{h-1}+1)q^h + (2a_{h-1}-q+1)\sum_{i=0}^{h-1}q^i $$
It holds if and only if $a_h+a_{h-1}\ge q-1$. \\

Case 4 : For $j\ge h+2$, it is easy to show that $q^ja\mn > a$ for  all $a$.\\

Summarizing the four cases above, we conclude that $a$ is not a coset leader in the following cases.

1) $a=a_h\sum_{i=1}^h q^i+q+a_0$, with $q+a_0-2a_h\le 0 $ or $q+a_0-2a_h=1$, $2a_h\ge q$;

2) $a=a_h\sum_{i=1}^hq^i+a_0$ with $a_h<a_0\le 2a_h$;

3) $a=a_hq^h+a_{h-1}\sum_{i=0}^{h-1}q^i+1$ with $a_h+a_{h-1}\ge q$ or $a_h+a_{h-1}=q-1$ while $2a_{h-1}\ge q$.
\end{proof}

\begin{corollary}
Let $q\ge 3$ and let $m$ be odd. Then the smallest positive $a \not\equiv 0\mq$ that is not a coset leader is $(q^{(m+1)/2}-1)/(q-1)+1$.
\end{corollary}

To study the dimensions of the BCH codes for $b\le -1$, we will analyse the set $\PCL^-$, which was defined in Section \uppercase\expandafter{\romannumeral2}. For any even $m$, $\PCL^-(-q^{m/2},q^{m/2})$ was given in Lemma 27 of \cite{Ding2016a}. Here we consider the case that $m$ is odd.

\begin{proposition} \label{propositionPC}
Let $m \ge 5$ be an odd integer. Set $h=(m-1)/2$. Define 
$$\PCL^-(-b_0,b_1)=\{\,x\in \CL(-b_0,-1) \mid  \exists j \text{ s.t. } (q^jx \mn)\in [1,b_1]\,\}$$ 
for positive integers $b_0,\, b_1$. Then we have
\begin{eqnarray*}
\lefteqn{\PCL^-(-q^{h+1},q^{h+1}) =} \\ 
& & \left\{\,a_hq^h+a_{h-1}\frac{q^h-1}{q-1}\mid  0\le a_h,a_{h-1}\le q-1,~a_{h-1}\neq 0,\,a_h\neq a_{h-1}\, \right\}\\
& \cup & \left\{\,a_h\frac{q^{h+1}-q}{q-1}+a_{h-1}-a_h\mid  0\le a_h,a_{h-1}\le q-1,~a_{h-1}\neq 0,\,a_h\neq a_{h-1}\,\right\} 
\end{eqnarray*}
and $|\PCL^-(-q^{h+1},q^{h+1})|=2(q-1)^2$.
\end{proposition}

\begin{proof}
For an integer $a\in [1,q^{h+1}]$, $-a\in \PCL^-(-q^{h+1},q^{h+1})$ is equivalent to
\begin{equation} \label{equationTC} -a\in \CL(-q^{h+1},-1) \text{ and } a+bq^j\equiv 0\end{equation} for some $b\in [1,q^{h+1}]$ and $1\le j\le m-1$. Then we check all integers $j$ with $1\le j\le m-1$ to search for the $a,b$ satisfying \eqref{equationTC}. Notice that $a+bq^j\equiv aq^{m-j}+b\equiv 0\pmod{n}$ is symmetric for $1\le j\le h$ and $h+1\le j \le m-1=2h$. Thus it suffices to consider $1\le j\le h$.

Denote the $q$-adic expansions of $a$ and $b$ by $a=\sum_{i=0}^ha_iq^i$ and $b=\sum_{i=0}^hb_iq^i$, respectively. When $1\le j\le h-1$, we have
\begin{align*}
a+bq^j <q^{h+1}+q^{2h}<n,
\end{align*}
which shows $a+bq^j\not \equiv 0\pmod{n}$.

When $j=h$, we have
\begin{align*}
a+bq^j\mn =&\sum_{i=h+1}^{m-2}(b_{i-h}-b_h)q^i+(b_0+a_h-b_h)q^h+\sum_{i=0}^{h-1}(a_i-b_h)q^i.
\end{align*}
Since $0\le a_i,b_i\le q-1$ and $a_0,b_0\neq 0$, it is easy to see that $a+bq^j\mn =0$ only in the following two cases:

\begin{itemize}
  \item $b_h=b_{h-1}=\cdots = b_1=a_0=\cdots=a_{h-1}$, $b_0+a_h=b_h$. This gives $a=a_hq^h+a_{h-1}(q^{h-1}-1)/(q-1)$ and $b=a_{h-1}(q^{h+1}-1)/(q-1)+a_{h-1}-a_h$ for $0\le a_h<a_{h-1}\le q-1$.
  \item $b_h=b_{h-1}=\cdots=b_2=a_0=\cdots=a_{h-1}$, $b_1=b_h-1$ and $b_0+a_h=q+b_h$. This gives $a=a_hq^h+a_{h-1}(q^{h-1}-1)/(q-1)$ and $b=a_{h-1}(q^{h+1}-1)/(q-1)+a_{h-1}-a_h$ for $1\le a_{h-1}<a_h\le q-1$.
\end{itemize}

By Proposition \ref{thmOddMProjective}, we see that all of them are coset leaders, which implies they are all in $\PCL^-$.
Combining these two cases we get the set $\PCL^-(-q^{h+1},q^{h+1})$.
\end{proof}

In particular for $\delta=kq^h+1$ with $1\le k\le q$, the set $\PCL^-(1-\delta,\delta-1)$ is determined in the following corollary.

\begin{corollary}\label{corProjectivePCL-}
Let $m \ge 5$ be an odd integer. Set $h=(m-1)/2$. Define 
$$\PCL^-(-b_0,b_1)=\{~x\in \CL(-b_0,-1)\mid  \exists j \text{ s.t. } (q^jx \mn)\in [1,b_1]~\}$$ 
for positive integers $b_0,\, b_1$. Then for $\delta=kq^h+1$ with  $1\le k\le q$, we have
\begin{eqnarray*}
\lefteqn{\PCL^-(1-\delta,\delta-1)} \\      
&=& \left\{~a_hq^h+a_{h-1}\frac{q^h-1}{q-1}\mid  0\le a_h,a_{h-1}\le k-1,~a_{h-1}\neq 0,\,a_h\neq a_{h-1}\, \right\}\\
& \cup & \left\{~a_h\frac{q^{h+1}-q}{q-1}+a_{h-1}-a_h\mid 0\le a_h,a_{h-1}\le k-1,~a_{h-1}\neq 0,\,a_h\neq a_{h-1}\,\right\}
\end{eqnarray*}
and $|\PCL^-(1-\delta,\delta-1)|=2(k-1)^2$.\end{corollary}
\begin{proof} The proof is very similar to that of Proposition \ref{propositionPC} and is omitted here. \end{proof}

\subsection{Projective BCH codes over GF($q$) with $\delta+b-2 \le q^{\lceil \frac{m}{2}\rceil}$}

A cyclic code is called \textit{reversible} if its generator polynomial $g(x)$ is self-reciprocal, 
i.e., $g(x)$ is equal to its reciprocal.  

When $m$ is even, the dimensions of the reversible narrow-sense projective BCH codes of length $n=(q^m-1)/(q-1)$ were settled in \cite{Ding2016a} for $\delta < q^{\lceil \frac{m}{2}\rceil}$. 
In this subsection, we consider the case that $m \ge 3$ is odd and $\delta$ is in the same range.

\subsubsection{Narrow-sense projective BCH codes when $m$ is odd}

Denote $h=(m-1)/2$. When $\delta\le q^h$, since all integers $a\in [1,q^h]$ with $a \not\equiv 0~\mq$ are coset leaders and $|C_a|=m$, the dimension of $\C_{(n,q,\delta,b)}$ is equal to $n-m\delta_{Nq}$ , where $\delta_{Nq}=\delta-1-\lfloor(\delta-1)/q\rfloor$. Next we assume that $q^h+1\le \delta\le q^{h+1}+1$. For simplicity, we let $\delta = kq^h+1$ for $1\le k\le q$.

\begin{theorem}\label{thmProjevtiveDimension}
Let $m \ge 5$ be an odd integer. Set $h=(m-1)/2$. For $b=1$ and $\delta = kq^h+1$ with $1\le k\le q-1$, the dimension of $\C_{(n,q,\delta,1)}$ is given by 
\begin{equation*}
\dim(\C_{(n,q,\delta,1)})=\begin{cases}
n-m\left(\delta_{Nq}-k(k-1)\right),~&\text{if }k\le\lfloor\frac{q}{2}\rfloor;\\
n-m\left(\delta_{Nq}-k(k-1)+2k-q\right),~&\text{if }\lfloor\frac{q}{2}\rfloor+1\le k\le q-1;\\
n-m\left(\delta_{Nq}-k(k-1)+2k-2\right),~&\text{if }k=q,\\
\end{cases}
\end{equation*}
where $\delta_{Nq}=\delta-1-\lfloor\frac{\delta-1}{q}\rfloor$ .
\end{theorem}
 \begin{proof}
We need to investigate the set
 $$\NCL:=\{a : a \in [1,\delta-1], q \nmid a, \text{ and $a$ is not a coset leader}\}.$$ By Proposition \ref{thmOddMProjective},  for $\delta=kq^h+1$, we can divide the set
 $\NCL$ into following three subsets.
$$\NCL_1=\left\{a_h \frac{q^{h+1}-q}{q-1}+q+a_0\mid  q+a_0-2a_h\le 1\, , q/2\le a_h\le k-1\right\},$$
$$\NCL_2=\left\{a_h \frac{q^{h+1}-q}{q-1}+a_0\mid  1\le a_h<a_0\le 2a_h<2k\right\}, $$
$$\NCL_3=\left\{a_hq^h+a_{h-1} \frac{q^{h}-1}{q-1} +1\mid  1\le a_h\le k-1,\, a_{h-1}<q-1,\,a_h+a_{h-1}\ge q\right\}$$
$$\cup\left\{a_hq^h+a_{h-1} \frac{q^{h}-1}{q-1}+1\mid  1\le a_h\le k-1\,, \frac{q}{2}\le a_{h-1}<q-1\,,a_h+a_{h-1}=q-1 \right\}.$$
Then it is easy to see that $\NCL_1\cap \NCL_2=\NCL_1\cap \NCL_3=\emptyset$.

Next we calculate the cardinalities of these sets. If $q$ is odd, we let $\bar{q}=(q-1)/2$.

When $k\le \bar{q}$, it is clear that we have $|\NCL_2|=|\NCL_3|=k(k-1)/2$ and $|\NCL _1|=0$. And we have $|\NCL_2\cap \NCL_3|=0$.

When $\bar{q}+1\le k\le q-1$, we have $|\NCL_1|=(k-\bar{q})(k-\bar{q}-1)$, $|\NCL_2|=(q-1)^2/4-(q-1-k)(q-k)/2$, and $|\NCL_3|=(k-1)(k-2)/2+\bar{q}-1$. In addition, $|\NCL_2 \cap \NCL_3|= k-\bar{q}-1$. Thus
$$|\NCL|=|\NCL_1|+|\NCL_2|+|\NCL_3|-|\NCL_2\cap \NCL_3|=k^2-3k+q.$$

When $k=q$, we similarly have $|\NCL|=q^2-3q+2$.

By \eqref{EqDim}, the dimension then follows directly. If $q$ is even, the desired results can be similarly obtained and the proof is omitted here.
\end{proof}
 
\subsubsection{Reversible projective BCH codes when $m$ is odd}
 
In this subsection, we study the reversible projective BCH codes $\C{(n,q,2\delta,1-\delta)}$. For simplicity, we also study the case that $\delta=kq^h+1$, where $1\le k\le q$.

\begin{theorem}\label{thm-ding161}
Let $m \ge 5$ be an odd integer. Set $h=(m-1)/2$. For $1 \le k \le q$ and $\delta = kq^h+1$, the dimension of $\C_{(n,q,2\delta,1-\delta)}$ is given by 
 \begin{eqnarray*}
\lefteqn{\dim(\C_{(n,q,2\delta,1-\delta)})=} \\ 
& 
\begin{cases}
n-1-2m\left(\delta_{Nq}-(2k-1)(k-1)\right),~&\text{if }k\le\lfloor\frac{q}{2}\rfloor;\\
n-1-2m\left(\delta_{Nq}-(2k-1)(k-1)+2k-q\right),~&\text{if }\lfloor\frac{q}{2}\rfloor+1\le k\le q-1;\\
n-1-2m\left(\delta_{Nq}-(2k-1)(k-1)+2k-2\right),~&\text{if }k=q,\\
\end{cases}
\end{eqnarray*}
where $\delta_{Nq}=\delta-1-\lfloor\frac{\delta-1}{q}\rfloor$ .
 \end{theorem}
 \begin{proof}
The conclusion follows directly from Corollary \ref{corProjectivePCL-} and Theorem \ref{thmProjevtiveDimension}.
 \end{proof}
 
\section{The case that $n=q^m+1$}

BCH codes with length $n=q^m+1$ are always reversible cyclic codes. 
In this section, we study the dimensions of the BCH codes of length $q^m+1$.
We also discuss the coset leaders before analysing the parameters of the BCH codes. 
Throughout this section, $n=q^m+1$ unless otherwise stated. 

\subsection{Auxiliary results about $q$-cyclotomic cosets modulo $n$}

For $n=q^m+1$, it was proved in   \cite{Ding2016a} that an integer $a$ is a coset leader with $|C_a|=2m$ for all $1\le a\le q^{\lfloor \frac{m-1}{2}\rfloor}+1$ when $a \not\equiv 0 \mq$. Below we consider $a$ in the  range $q^{\lfloor \frac{m-1}{2}\rfloor}+1 \le a \le q^{\lfloor \frac{m+1}{2}\rfloor}$.

When $m$ is even, we have the following conclusion.

\begin{proposition}\label{thmIV1}
Let $m$ be an even integer with $m \ge 4$. Set $h=m/2$. For $q^{h-1} \le a \le q^h$ with 
$a\not\equiv 0\mq$, $a$ is a coset leader with $|C_a|=2m$.
\end{proposition}

\begin{proof}
Let $\sum_{i=0}^h a_iq^i$ be the $q$-adic expansion of $a$.
When $q^{h-1}\le a \le q^h$ and $a\not\equiv 0\mq$, we have $a_{h-1}\neq 0$ and $a_0 \neq 0$.

When $1\le j\le h$, we clearly have $a<q^ja<n$.

When $h+1\le j \le 2h-1$, we have
\begin{align*}
q^ja\mn &= \sum_{i=j}^{l-1}a_{i-j}q^i-\sum_{i=0}^{j-h-1}a_{i+2h-j+1}q^i
\ge a_0q^j - \sum_{i=0}^{h-2}(q-1)q^i
\ge q^h
>a.
\end{align*}

When $2h\le j\le 3h$, we have $$q^ja\mn=n-q^{j-2h}a\ge q^2h-q^h*(q^h-1)+1=q^h+1>a.$$

When $3h+1\le j\le 4h-1$, let $j' = j - 2h$ we have $h+1\le j' \le 2h-1$ and
\begin{align*}
q^ja\mn
=n-\sum_{i=j'}^{l-1}a_{i-j'}q^i +\sum_{i=0}^{j'-h-1}a_{i+2h-j'}q^i
\ge q^{h+1}
>a.
\end{align*}
Then the desired conclusion follows.
\end{proof}

When $m$ is odd, we have the following conclusion.

\begin{proposition}\label{thmIV2}
Let $m\ge 3$ be an odd integer. Set $h=(m-1)/2$. For $q^h\le a\le a^{h+1}$ with $a\not\equiv 0\mq$,   $a$ is a coset leader with $|C_a|=2m$ except that $a=q^{h+1}-c$ for $1\le c\le q-1$.
\end{proposition}

\begin{proof}
Let $\sum_{i=0}^h a_iq^i$ be the $q$-adic expansion of $a$.
For $q^{h}\le a \le q^{h+1}$ and $a\not\equiv 0\mq$, we have $a_{h}\neq 0$ and $a_0 \neq 0$.

When $1\le j \le h$, $a<q^ja<n$.

When $j=h+1$, we have
\begin{align*}
q^ha\mn = \sum_{i=h+1}^{2h}a_{i-h-1}q^i - a_h.
\end{align*}
If one of $a_0-1,a_1,...,a_{h-1}$ is nonzero, we have $q^ja\mn>q^{h+1}>a$. Otherwise, we have 
$a= a_hq^h+1\le (q-1)q^h +1$, and $q^ja\mn  \ge q^{h+1}-q>a$.

When $h+2\le j\le 2h$, we have
$$q^ja\mn = \sum_{i=j}^{2h}a_{i-j}q^i - \sum_{i=0}^{j-h-1}a_{i+2h-j+1}q^j\ge q^{h+1}>a.$$

When $2h+1\le j\le 3h$, set $j'=j-2h-1$, we have $j'\le h-1$ and
\begin{align*}
q^ja\mn =
n-\sum_{i=j'}^{j'+h}a_{i-j'}q^i\mn
\ge q^{2h}
&>a.
\end{align*}

When $j=3h+1$, we have
\begin{align*}
q^ja\mn = -\sum_{i=h}^{2h}a_{i-h}q^i\mn = n-\sum_{i=h}^{2h}a_{i-h}q^i.
\end{align*}

If one of $a_1, a_2, \ldots, a_h$ is equal to $q-1$, let $a_k$ be the first one that satisfies 
$a_k\le q-2$. Then
$$q^ja\mn\ge (q-1-a_k)q^{h+k}\ge q^{h+1}>a.$$

Otherwise,  we can see from $a_0\neq 0$ that
$$q^ja\mn = (q-1-a_0)q^h +1 < (q-1)\sum_{i=1}^h q^i+a_0 = a.$$

Therefore $q^ja\mn < a $ if and only if $a=(q-1)\sum_{i=1}^h q^i+a_0$ with $1\le a_0\le q-1.$

When $3h+2\le j\le 4h-1$, set $j' = j - 2h-1$, we have
\begin{align*}
q^ja\mn =
 n-\sum_{i=j}^{2h}a_{i-j}q^i + \sum_{i=0}^{j-h-1}a_{i+2h-j+1}q^j
\ge q^{h+1}
>a.
\end{align*}
Concluding all above, $a$ is a coset leader if and only if $a=\sum_{i=1}^h (q-1)q^i+a_0$ where $1\le a_0\le q-1$.
\end{proof}

\begin{corollary}
For odd $m \ge 3$, the smallest $a\not\equiv 0\mq$ that is not a coset leader is $q^{(m+1)/2}-q+1$. 
\end{corollary}

\subsection{ BCH codes over GF($q$) with $n=q^l+1$ and $\delta+b-2 \le q^{\lceil \frac{m+2}{2}\rceil}$}

With the results on the cyclotomic cosets in the range $[1,q^{\lfloor \frac{m+1}{2}\rfloor}]$ 
developed above, we have the following conclusions on parameters of BCH codes with $n=q^m+1$. 
Their proofs follow directly from Propositions \ref{thmIV1} and \ref{thmIV2} and are omitted.

\begin{theorem}\label{thm-ding162}
Let $m \ge 4$ be an even integer, and let $h=m/2$. Then for $2\le \delta \le q^h$, 
the narrow-sense BCH code $\C_{(n,q,\delta,1)}$ has parameters
$$\left[\,q^m+1,\,q^m+1-2m\left(\delta-1-\left\lfloor\frac{\delta-1}{q}\right\rfloor\right),\,d\ge \delta\,\right]$$
and $\C_{(n,q,\delta+1,0)}$ has parameters
$$\left[\,q^m+1,\,q^m-2m\left(\delta-\left\lfloor\frac{\delta}{q}\right\rfloor\right),\,d\ge 2\delta\,\right].$$
\end{theorem}

\begin{theorem}\label{thm-ding163}
Let $m \ge 3$ be an odd integer and let $h=(m-1)/2$. Then for $2\le \delta \le q^{h+1}$ we have
$$\dim(\C_{(n,q,\delta,1)})=\begin{cases}
q^m+1-2m(\delta-1-\lfloor\frac{\delta-1}{q}\rfloor),~&\text{if }\delta \le q^{h+1}-q;\\
q^m+1-2m(q^{h+1}-q-\lfloor\frac{\delta-1}{q}\rfloor),~&\text{if }q^{h+1}-q+1\le \delta\le q^{h+1}.
\end{cases}
$$
\end{theorem}

\section{Conclusions and remarks}

In this paper, we mainly investigated the dimensions of the BCH codes $\C_{(n,q,\delta,b)}$ for 
three types of lengths, i.e., $n=q^m-1$, $n=(q^m-1)/(q-1)$ and $n=q^m+1$. We explored the 
dimensions for different $b$ and $\delta$. 
In addition, we extended a known result on the minimum distances of narrow-sense BCH codes 
and applied it to several BCH codes, whose parameters were therefore completely settled 
(see Lemma \ref{lemmaMinDis} and Theorem \ref{thm-ding151}).

For the primitive narrow-sense BCH code $\C_{(q^m-1,q,\delta,1)}$, we settled its dimension 
for all $\delta$ with $1\le \delta \le q^{\lceil(m+2)/2\rceil}$. This extends earlier work to 
a large extent. However, the dimension of this code is still unknown for  $\delta > q^{\lceil(m+2)/2\rceil}$, except for a few special $\delta$ in this range.    
For the non-narrow-sense cases (i.e., $b \neq 1$), we derived several dimension formulas for the 
code $\C_{(q^m-1, q, \delta,b)}$ in general and determined the dimension of this code for some 
specific types of $\delta$. Our results about the dimension of primitive BCH codes are documented 
in Theorems \ref{thm-sum-primBCH1},
\ref{thmPrimitiveMeven}, 
\ref{thmDimensionModdPrimitve}, 
\ref{thm-ding151}, 
\ref{thm-sum-primBCH2}, 
\ref{thm-sum-primBCH3}, and 
\ref{thm-sum-primBCH4}. 
Although most of the references on BCH codes dealt with the primitive 
case, the dimension of most of the primitive BCH codes is unknown, let alone their minimum distances.     

It might be true that \cite{Ding2016a} and \cite{Lia} are the only references on projective 
BCH codes of length $n=(q^m-1)/(q-1)$. In this paper, we settled the dimension of the projective 
BCH code $\C_{((q^m-1)/(q-1), q, \delta,b)}$ for odd $m$ and some special values of $\delta$ 
(see Theorems \ref{thmProjevtiveDimension} and \ref{thm-ding161}). Our result on the dimension 
of the projective BCH codes complements Theorem 29   
of \cite{Ding2016a}. It should be noticed that the dimension and minimum distance of the projective 
BCH code $\C_{((q^m-1)/(q-1), q, \delta,b)}$ are still open in general.  

The only published paper on the BCH codes $\C_{(q^m+1, q, \delta,b)}$ is \cite{Kenza}, where 
the dimension of $\C_{(q^m+1, q, \delta, 1)}$ is determined for $2 \leq \delta \leq q$. 
The dimension 
of $\C_{(q^m+1, q, \delta, 0)}$ was worked out for $3 \le \delta \le q^{\lfloor (m-1)/2 \rfloor}$ 
in \cite{Ding2016a}. In this paper, we calculated the dimension of $\C_{(q^m+1, q, \delta, 1)}$ 
for a larger range of $\delta$ (see Theorems \ref{thm-ding162} and \ref{thm-ding163}). Clearly, 
the parameters of the BCH code $\C_{(q^m+1, q, \delta,b)}$ are open in general. 

BCH codes $\C_{(n, q, \delta,b)}$ of many other types of lengths are untouched. For example, 
BCH codes of length $n=(q^m+1)/(q+1)$ are not investigated in the literature, where $m$ is 
odd. Hence, most of the BCH codes are not studied. The reader is thus cordially invited to 
uncover the world of BCH codes over finite fields.     

\section*{Acknowledgements}

C. Ding's research was supported by the Hong Kong Research Grants Council, Proj. No. 16300415.   




\end{document}